\documentclass[letterpaper]{article}
\usepackage{times} \usepackage{helvet} \usepackage{courier}
\usepackage{graphicx} \usepackage{subfig} \usepackage{amsfonts}
\usepackage{amsmath} \usepackage{multirow} \usepackage{amssymb}
\usepackage{amsthm}  
\newtheorem{theorem}{Theorem} 
\newtheorem{lemma}{Lemma} 
\newtheorem{definition}{Definition} \newtheorem{observation}{Observation}

\begin{document} 
\title{Towards Asymptotically Optimal One-to-One PDP Algorithms for Capacity 2+ Vehicles}
\author{Martin Olsen\\BTECH, Aarhus University\\Denmark\\martino@btech.au.dk\\ }

\maketitle

\begin{abstract} 
We consider the one-to-one Pickup and Delivery Problem (PDP) in
Euclidean Space with arbitrary dimension $d$ where $n$
transportation requests are picked i.i.d. with a separate
origin-destination pair for each object to be moved. First, we
consider the problem from the customer perspective where the
objective is to compute a plan for transporting the objects such
that the Euclidean distance traveled by the vehicles {\em when
carrying objects} is minimized. We develop a polynomial time
asymptotically optimal algorithm for vehicles with capacity
$o(\sqrt[2d]{n})$ for this case. This result also holds imposing
LIFO constraints for loading and unloading objects. Secondly, we
extend our algorithm to the classical single-vehicle PDP where the
objective is to minimize the total distance traveled by the vehicle
and present results indicating that the extended algorithm is
asymptotically optimal for a fixed vehicle capacity if the origins
and destinations are picked i.i.d. using the same distribution.
\end{abstract}

\section{Introduction}

The challenge of computing optimal or near-optimal plans for
transporting goods or people is a core challenge within logistics.
This problem has received a huge amount of attention from the
operations research community. A generic term for this class of
problems is {\em Vehicle Routing Problems}. Vehicle routing problems
come in many flavors depending among other things on the properties
of the vehicles used, the characteristics of the terrain and the
type of transportation requests considered.

In this paper, we consider the variant of the problem where the
terrain is the Euclidean space of an arbitrary dimension $d$ and
where we measure the distance using the Euclidean distance. We have
one vehicle with limited capacity at our disposal but the vehicle
can carry more than one object at a time. Every object to be
transported has a separate origin and destination. The objective is
to compute a plan for transporting the objects with a minimum
distance traveled by the vehicle. We look at the nonpreemptive
version meaning that an object has to stay on the vehicle until it
is delivered. Our version is static (offline) in the sense that all
information on the transportation requests is available to us before
we have to compute the optimal route for the vehicle.

The origins and destinations are picked using a stochastic process
and we measure the performance of an algorithm by considering the
{\em approximation ratio}, i.e. the value of the solution computed
by the algorithm divided by the value of the optimal solution. The
main aim of the paper is to present polynomial time algorithms that
are {\em asymptotically optimal} in the sense that the approximation
ratio converges to $1$ with probability $1$ (almost surely) as the
number of transportation requests goes to infinity. To the best of
our knowledge, we are the first to present asymptotically optimal
algorithms for the realistic setting where the vehicles have a
limited capacity greater than one. Our algorithms
are easy to implement and they may be useful in practice.

\subsection{Related Work}

In the literature, the problem considered is often referred to as
the One-to-One Pickup and Delivery Problem (PDP) or the Vehicle
Routing Problem with Paired Pickups and Deliveries. The single
vehicle case that we look at in this paper is also known as the
Traveling Salesman Problem with Pickups and Deliveries because this
case can be viewed as a Traveling Salesman Problem (TSP) with
precedence constraints where the origin of an object must be visited
before the corresponding destination. If the single vehicle has
capacity $1$ then this problem is known as the Stacker Crane Problem
(SCP). We refer the reader to the excellent surveys
~\cite{Berbeglia2007,Parragh2008_1,Parragh2008_2,Savelsbergh95} for
an overview on vehicle routing research.

The PDP problem in focus in this paper is defined as follows:
\begin{definition} 
An instance of the {\em PDP} problem for dimension $d$ and capacity
$c$ is a set of $n$ requests $R = \{r_1, r_2, \ldots , r_n\}$ with
$r_i \in [0,1]^d \times [0,1]^d$ for $i \in \{1, 2, \ldots , n\}$. A
request $r = (s,t)$ corresponds to a transportation job for an
object with origin $s$ and destination $t$. The solution is a plan
for transporting all objects from their origin to their destination
minimizing the total Euclidean distance traveled using a single
vehicle with capacity $c$. The vehicle is initially at
$(0,0,\ldots,0)$.
\end{definition}

As noted by many authors, the PDP problem generalizes the classical
TSP problem and is thus NP complete for $d \geq 2$.
Guan~\cite{GUAN1998} has shown that the PDP problem is NP complete
for $d=1$ and $c \geq 2$ and Guan also shows how to solve this
version in linear time if we allow temporarily dropping objects.
Treleaven et al.~\cite{Treleaven2013} present asymptotically optimal
algorithms for $c=1$ (SCP) and $d \geq 2$ where the origins are
picked i.i.d. and the destinations are picked i.i.d. from separate
distributions.

Stein~\cite{Stein1978} also conducts a probabilistic analysis but he
looks at the variant where $d=2$ and $c = +\infty$ and where $n$
origins and $n$ destinations are picked independently using a
uniform distribution on some planar region. Stein shows that the
value of the optimal solution divided by $\sqrt{n}$ converges almost
surely to a constant times the square root of the area of the
region. Stein also shows how to solve the problem he considers by
concatenating two TSP tours on the origins and destinations
respectively. This way of solving the problem yields a solution
which is roughly $6\%$ higher than the optimal solution (in the
limit).

Psaraftis~\cite{Psaraftis1981} has developed an $O(n^2)$ heuristic
guaranteeing an approximation ratio of at most $4$ for any instance
for the case $d=2$ and $c=+\infty$. Haimovich and
Kan~\cite{Haimovich_Kan_1985} have constructed asymptotically
optimal PDP algorithms for the case $d=2$ where all transportation
requests have the same depot as destination. Haimovich and
Kan~\cite{Haimovich_Kan_1985} also presented a PTAS for this case
for $c = O(\log \log n)$ and this result has later been extended to
cover cases with multiple depots and arbitrary values of
$d$~\cite{Khachay2016} or larger values of $c$~\cite{Das15}.

In some cases, the customers only pay for a kilometer driven by the
vehicle if the vehicle carries an object when driving that
kilometer. As an example, this is the case when the vehicles are
taxis. This leads us to the PDP-C problem that does not seem to have
received much attention:
\begin{definition} 
The {\em PDP-C} problem is identical to the PDP problem with the
exception that distance traveled {\em carrying no objects} is
excluded.
\end{definition}

The PDP-C problem covers any situation where
carrying objects is very expensive compared to carrying no objects.
Possible areas for application are the development of taxi sharing
or ride sharing schemes. The PDP-C problem also comes into play when
we want to minimize the time spent for an elevator (or robot arm)
that moves slowly carrying passengers (objects) but moves very fast
carrying no passengers (objects).

An efficient and near-optimal
subroutine solving the PDP-C problem might also be useful in the
case where we have multiple vehicles at our disposal and the
objective is to minimize the completion time (the time when the last
object has been delivered). A first step to solve this problem could
involve partitioning objects into groups sharing a vehicle using a
PDP-C subroutine.

\subsection{Contribution and Outline}

In Section~\ref{sec:pdp-c}, we present an adaptive asymptotically optimal
polynomial time algorithm for the PDP-C problem for $c =
o(\sqrt[2d]{n})$ for any dimension $d$ under the assumption that the
transportation requests are picked independently and by identical
distributions (i.i.d.). As explained above, there are many real
world problems where a PDP-C algorithm is useful.

A PDP algorithm is presented in Section~\ref{sec:pdp} accompanied by
what we consider to be good reasons to believe that this algorithm
is asymptotically optimal for a constant capacity $c$ if the origins
and destinations for the transportation requests are picked i.i.d.
using the same distribution. The PDP algorithm is an extension of
our PDP-C algorithm.

The key idea for our approach is that we solve a TSP with precedence
constraints in $d$-dimensional Euclidean space -- the PDP problem --
by solving a classical TSP defined by the requests in Euclidean
space with the double dimension $2d$. Neighbors in the solution to
the classical TSP correspond to objects that have similar requests
so we let such neighbors share a vehicle.

As mentioned earlier, we believe that we are the first to present
results on asymptotic optimality for the realistic case $1 < c <
+\infty$. Our results also hold if we allow temporarily dropping
objects (the preemptive variant) or if we impose LIFO constraints
for loading and unloading objects.


\section{An Asymptotically Optimal PDP-C Algorithm}\label{sec:pdp-c}

In this section, we present our PDP-C algorithm. We begin by listing the pseudocode consisting of $4$ steps. We also exemplify how the steps work using the instance shown in Fig.~\ref{fig:instance} that has $d=1$ and $c=2$.

\subsection{The Pseudocode}
\begin{figure}
  \centering 
  \includegraphics[scale=0.9]{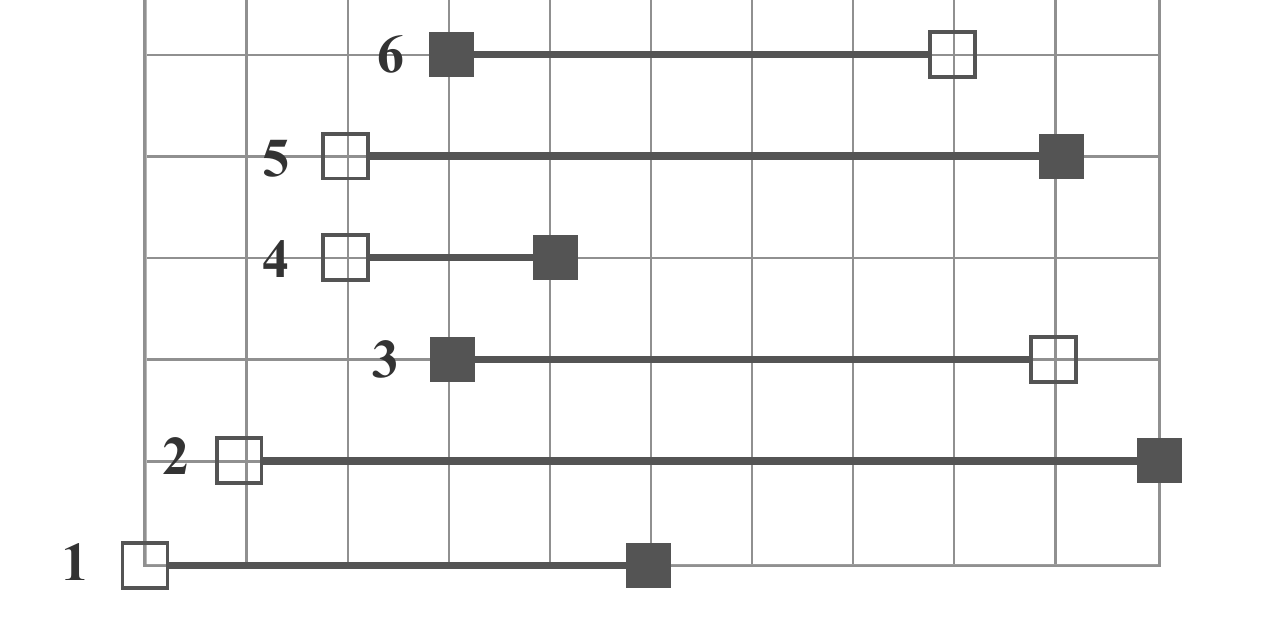}
  \caption{An instance of a  $1$-dimensional PDP-C problem. The white and black squares correspond to pickups and deliveries respectively. As an example, request $4$ shows that an object has to be picked up at $0.2$ and delivered at $0.4$. We assume that vehicles with capacity $2$ are used.}
  \label{fig:instance}
\end{figure}

Throughout the paper, we assume that $n = mc$ for some integer $m$. Otherwise, we can serve the extra objects one by one implying an extra $O(c)$ cost that does not affect our results on asymptotic optimality. 

\subsubsection*{Step $1$}
Use a polynomial time constant factor approximation algorithm~\cite{Arora1998,christofides76} to compute a feasible solution $T$ for the $2d$-dimensional
Euclidean TSP problem defined by $R$ ($\sigma$ is a permutation on $\{1, 2, \ldots, n\}$): 
$$T = r_{\sigma(1)} \rightarrow r_{\sigma(2)} \rightarrow \ldots \rightarrow r_{\sigma(n)} \rightarrow r_{\sigma(1)} \enspace .$$ The tour $T$ is shown in Fig.~\ref{fig:T}.
\begin{figure}
  \centering 
  \includegraphics[scale=0.9]{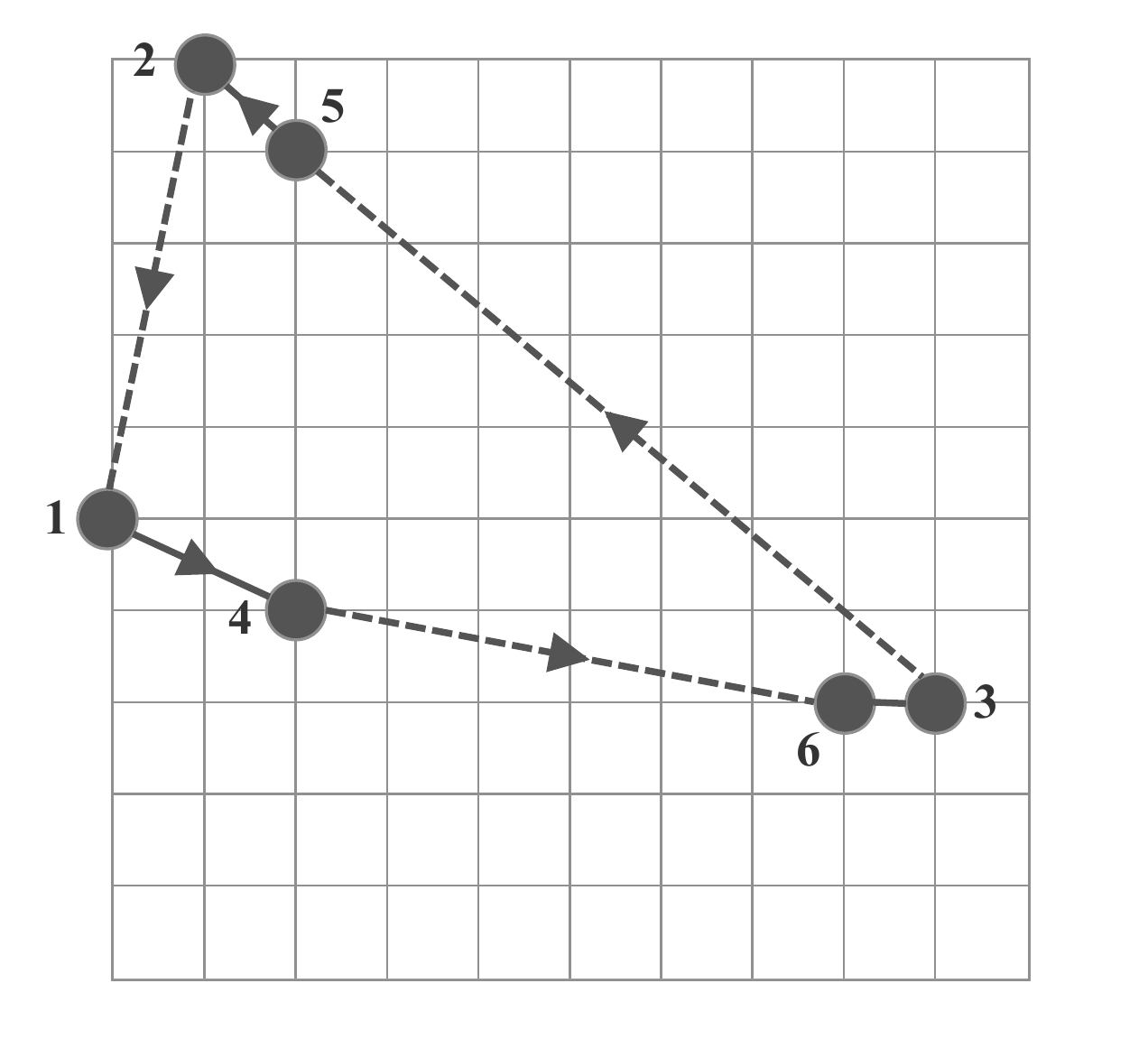}
  \caption{The objects that are to be transported live in $1$-dimensional Euclidean space. The requests from Fig.~\ref{fig:instance} are members of the $2$-dimensional Euclidean space. As an example, request $4$ is the point $(0.2,0.4) \in [0,1]^2$. The PDP-C algorithm attacks a $1$-dimensional PDP-C problem by solving a classical $2$-dimensional TSP defined by the requests. The two ways to split $T$ into groups are indicated by the dashed and solid arrows respectively.}
  \label{fig:T}
\end{figure}

\subsubsection*{Step $2$}
We now split $T$ into $m$ groups such that each group contains $c$
consecutive requests $r_{\sigma(i)}, r_{\sigma(i+1)}, \ldots ,
r_{\sigma(i+c-1)}$. One possible way of partitioning the requests into groups is
as follows:
$$\{r_{\sigma(1)}, r_{\sigma(2)}, \ldots, r_{\sigma(c)}\},$$
$$\{r_{\sigma(c+1)}, r_{\sigma(c+2)}, \ldots, r_{\sigma(2c)}\},$$
$$ \ldots $$
$$\{r_{\sigma(n-c+1)}, r_{\sigma(n-c+2)}, \ldots, r_{\sigma(n)}\}\enspace .$$
There are $c$ ways to do the split up (See Fig.~\ref{fig:T}). For
each of the possibilities, we repeat Step $3$ and obtain $c$ candidate solutions for the PDP-C problem:

\subsubsection*{Step $3$ (repeated for each possible splitting of $T$)}
The objects in a group share a vehicle. The objects for a group of requests $\{r_{\sigma(i)}$, $r_{\sigma(i+1)}$, $\ldots$ , $r_{\sigma(i+c-1)}\}$ are picked up in the order $\sigma(i)$, $\sigma(i+1)$, \ldots , $\sigma(i+c-1)$ and dropped off in reverse order (LIFO). The plan corresponding to one of the ways to split $T$ into groups is shown in Fig.~\ref{fig:cabsharing}.

\subsubsection*{Step $4$}
Finally, we pick the best of the $c$ candidate solutions produced in Step $3$. The plan computed by the algorithm is shown in Fig.~\ref{fig:cabsharing}.

\begin{figure}
  \centering 
  \includegraphics[scale=0.9]{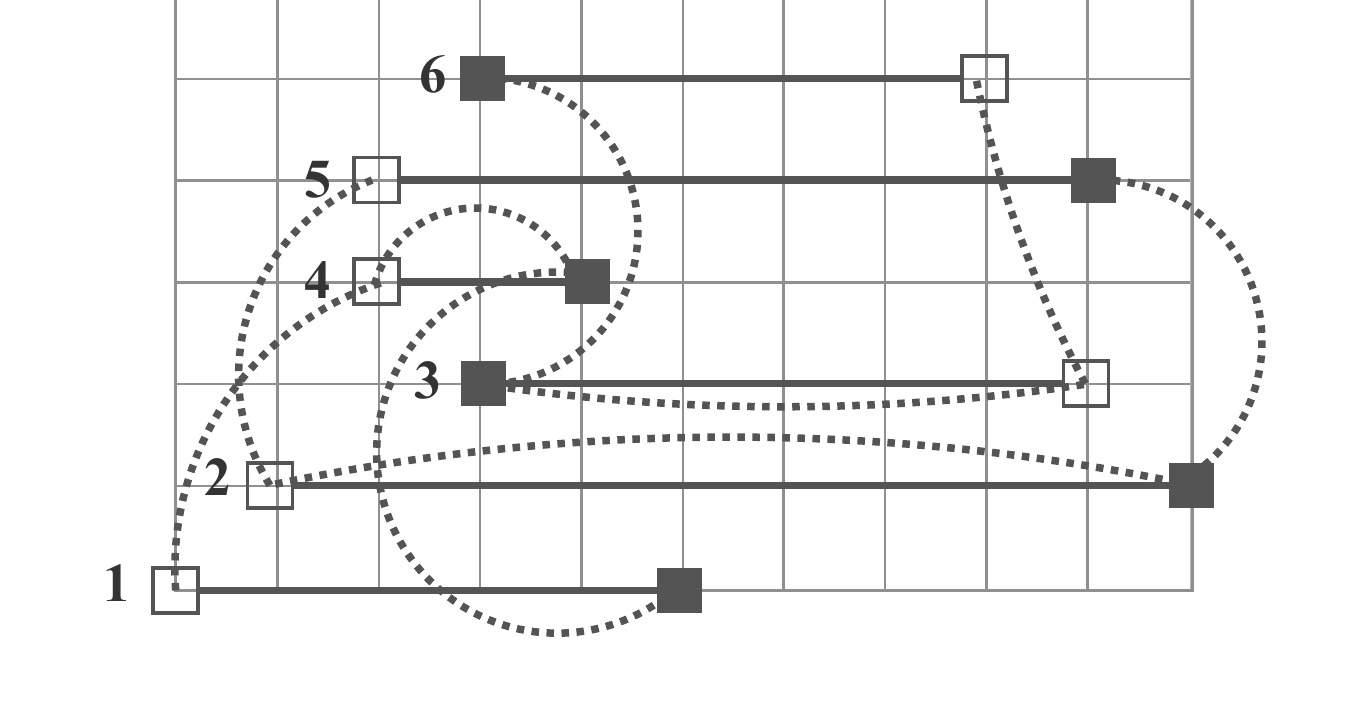}
  \caption{The feasible solution for the PDP-C instance computed by the PDP-C algorithm. The solution is based on the splitting of $T$ indicated by solid arrows in Fig.~\ref{fig:T}. The objects are picked up and dropped in LIFO order.}
  \label{fig:cabsharing}
\end{figure}

\subsection{Analysis of the PDP-C Algorithm}

We let $SOL$ denote the value of the plan computed by the PDP-C
algorithm. We kindly remind the reader that the value is the
Euclidean distance traveled where we only measure the distance
traveled carrying objects. The value of the optimal solution is
denoted by $OPT$. The length of the tour $T$ is $\|T\|_{2d}$ where
the subscript indicates the dimension of the underlying Euclidean
space.

We now present a key lemma that links the TSP in $2d$-space to the
value of the solution computed by the PDP-C algorithm:
\begin{lemma}
\label{lemma:cab}
\begin{equation}
\label{eq:lemma:cab}SOL \leq \sum_{i=1}^n \frac{\|s_i - t_i\|_d}{c} + \sqrt{2}\left(\frac{c-1}{c}\right) \|T\|_{2d} \enspace .
\end{equation}
\end{lemma}

\begin{proof}
Let $SOL_i$, $i \in \{1, 2, \ldots , c\}$, denote the Euclidean distance covered for the $i$'th plan computed by Step $3$. 
The sum $\sum_{i=1}^c SOL_i$ can be broken down into three terms: \begin{equation}\sum_{i=1}^c SOL_i = P+F+D \enspace \end{equation} where $P$ is a term for picking up objects, $F$ is a term for driving with a full vehicle, and $D$ is a term for dropping off objects.
Every object tries to be the final object to be picked up in precisely one of the plans:
\begin{equation}F = \sum_{i=1}^n \|s_i - t_i\|_d\enspace .\end{equation}
The segments $s_{\sigma(i)} \rightarrow s_{\sigma(i+1)}$ and $t_{\sigma(i+1)} \rightarrow t_{\sigma(i)}$ are traversed in exactly $c-1$ of the plans, $i \in \{1, 2, \ldots, n\}$ (Addition is performed cyclically: $n+1 = 1$): 
\begin{equation}P+D = (c-1) \sum_{i=1}^n \left[\|s_{\sigma(i)} - s_{\sigma(i+1)}\|_d +  \|t_{\sigma(i)} - t_{\sigma(i+1)}\|_d \right] \enspace .\end{equation}
The key to connecting the $d$-dimensional space to the $2d$-dimensional space is the following simple observation that follows from the elementary identity $\sqrt{x}+\sqrt{y} \leq \sqrt{2}\sqrt{x+y}$:
\begin{equation}\|s_{\sigma(i)} - s_{\sigma(i+1)}\|_d +  \|t_{\sigma(i)} - t_{\sigma(i+1)}\|_d \leq \sqrt{2}\|r_{\sigma(i)} - r_{\sigma(i+1)}\|_{2d} \enspace .\end{equation} 
We now combine the two preceding facts:
\begin{equation}P+D \leq \sqrt{2}(c-1) \|T\|_{2d}\enspace .\end{equation}
The Lemma now follows from
\begin{equation}SOL \leq \frac{1}{c}\sum_{i=1}^c SOL_i \enspace .\end{equation}
\end{proof}

We are now ready to prove that our PDP-C algorithm is asymptotically optimal:
\begin{theorem}
\label{thm:d+}
Let an infinite sequence of requests $(s_i,t_i)$ be
picked i.i.d. in $[0,1]^d \times [0,1]^d$ using a distribution satisfying that $E[\|s_i
- t_i\|_d] = \mu > 0$. Let $SOL_n$ denote the
value of the plan computed by the PDP-C algorithm and let
$OPT_n$ denote the value of the optimal plan for the first $n$ requests. If $c =
o(\sqrt[2d]{n})$ then we have the following:
\begin{equation} \lim_{n \rightarrow +\infty} \frac{SOL_n}{OPT_n} = 1 \mbox{ a.s.} \end{equation}
\end{theorem}

\begin{proof}
The objects could share the bill of traveling by equally sharing the cost for each segment. This sharing scheme leads to the following lower bound on $OPT_n$:
\begin{equation}OPT_n \geq \sum_{i=1}^n \frac{\|s_i - t_i\|_d}{c}\end{equation}
We now combine the lower bound on $OPT_n$ with Lemma~\ref{lemma:cab}:
\begin{equation}\label{eq:z} \frac{SOL_n}{OPT_n} \leq 1 + \sqrt{2}(c-1) \left(\sum_{i=1}^n \|s_i - t_i\|_d\right)^{-1} \|T\|_{2d} \enspace .\end{equation}
The inequality (\ref{eq:z}) is rewritten slightly:
\begin{equation}\label{eq:appx}\frac{SOL_n}{OPT_n} \leq 1 + \sqrt{2}(c-1) \left(\frac{\sum_{i=1}^n \|s_i - t_i\|_d}{n}\right)^{-1} \frac{\|T\|_{2d}}{n} \enspace .\end{equation}
We use a constant factor approximation algorithm for solving the TSP in $2d$-dimensional space implying this upper bound on the length of $T$~\cite{few1955}:
\begin{equation}\|T\|_{2d} = O(n^{\frac{2d-1}{2d}}) \enspace .\end{equation}
Using $c = o(\sqrt[2d]{n})$, we now get the following:
\begin{equation}\label{eq:ct}c\|T\|_{2d}= o(n) \enspace .\end{equation}
According to the Strong Law of Large Numbers, we have the following:
\begin{equation}\label{eq:stronglaw} \lim_{n \rightarrow +\infty} \left(\frac{\sum_{i=1}^n \|s_i - t_i\|_d}{n}\right)^{-1} = \mu^{-1}\mbox{ a.s.} \end{equation}
The lemma now follows from (\ref{eq:appx}), (\ref{eq:ct}) and (\ref{eq:stronglaw}). 
\end{proof}

A few comments on the convergence rate might be suitable at this
point. According to~\cite{BeardwoodHaltonHammersley1959}, the limit
of $\frac{\|T\|_{2d}}{n} \sqrt[2d]{n}$ is almost surely a constant
where the constant depends on the distribution of the requests with
maximum value for the uniform distribution. In other words, the
algorithm is {\em adaptive} in the sense that the right hand side of
(\ref{eq:appx}) tends to be smaller for big instances for the
non-uniform case.

Even for relatively small values of $n$, we might experience a right
hand side of (\ref{eq:appx}) that is relatively close to $1$. As an
example, we consider the case $d=1$ where we admittedly have the
best conditions for convergence. Few~\cite{few1955} has shown that
$\|T\|_{2} \leq \sqrt{2n}+\frac{7}{4}$ implying that the right hand
side of (\ref{eq:appx}) converges relatively quickly to $1$ for
moderate $c$  for the case $d=1$ if $\mu$ is not too small.

\section{A PDP Algorithm}\label{sec:pdp}

We now turn our attention to the PDP problem and present a
polynomial time algorithm that can be viewed as a generalization of
the {\em Iterated Tour Partition Heuristic} that Haimovich and
Kan~\cite{Haimovich_Kan_1985} presented for the case where $d=2$ and
all the destinations are identical. The tour that we consider is a
tour in $2d$-dimensional Euclidean space where the requests are
members and we allow different destinations for arbitrary $d$.

Our PDP algorithm is an extension of our PDP-C algorithm: First, we
figure out what objects should share the vehicle and establish a LIFO order
for pickups and deliveries (PDP-C with $1 < c < +\infty$). Secondly,
we set up a route for the vehicle focusing on the segments when it
carries no objects (PDP with $c = 1$. SCP in other words). We
exemplify our PDP algorithm by adding two more figures to the
PDP-C example. The pseudocode for the PDP algorithm consists of the
following $6$ steps:

\subsubsection*{Steps $1-4$}
Use the steps from the PDP-C Algorithm and compute a PDP-C solution. 

\subsubsection*{Step $5$}
Use an algorithm from the SPLICE class~\cite{Treleaven2013} to compute a feasible solution $S$ for the SCP instance defined by a pickup at the origin and a delivery at the destination for every object that was the first to be picked up by a vehicle (and consequently the last object to be dropped off) in the PDP-C solution. Let $S_0$ denote the set of segments that go \underline{from} a delivery \underline{to} a pickup from the solution $S$ to the Stacker Crane Problem. See Fig.~\ref{fig:SCP}.
\begin{figure}
  \centering 
  \includegraphics[scale=0.9]{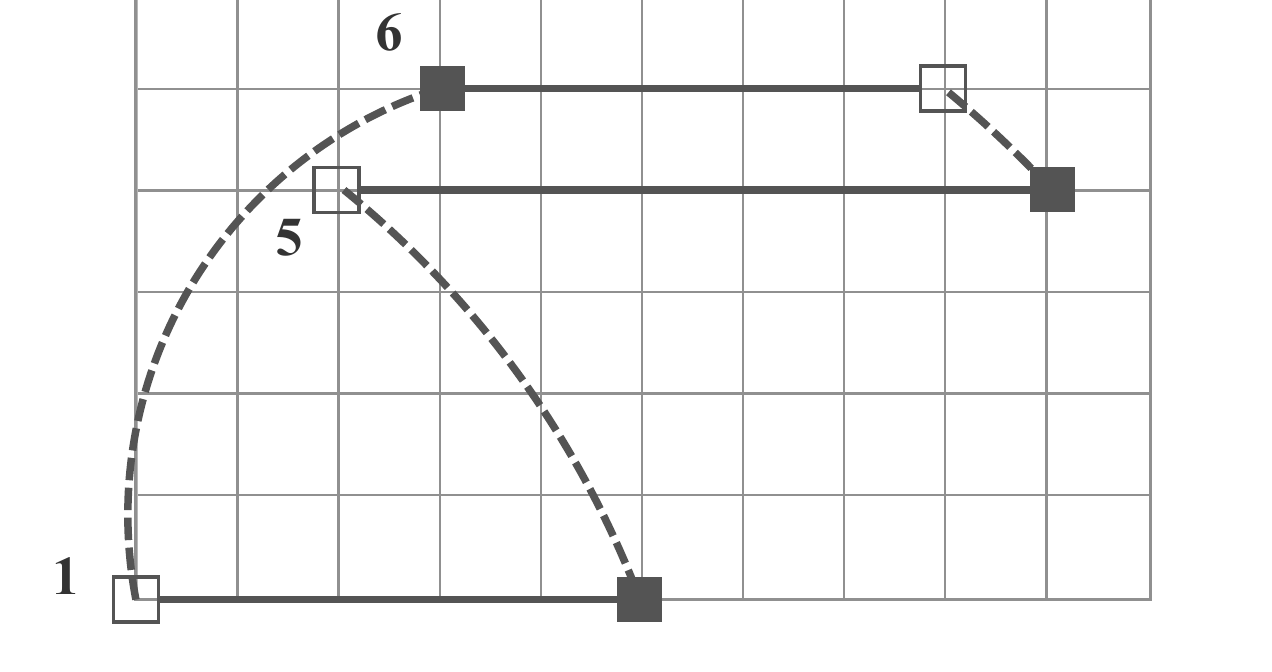}
  \caption{The SCP and the solution $S$ computed in Step $5$. The dashed segments are the segments in $S_0$. The objects for the requests $1$, $5$ and $6$ where the first to be picked up in the PDP-C solution in Fig.~\ref{fig:cabsharing}.}
  \label{fig:SCP}
\end{figure}

\subsubsection*{Step $6$}
A PDP solution can now be produced by combining the PDP-C solution with the segments $S_0$ from the Stacker Crane Plan where no objects where carried. See Fig.~\ref{fig:pdp}.\\
\begin{figure}
  \centering 
  \includegraphics[scale=0.9]{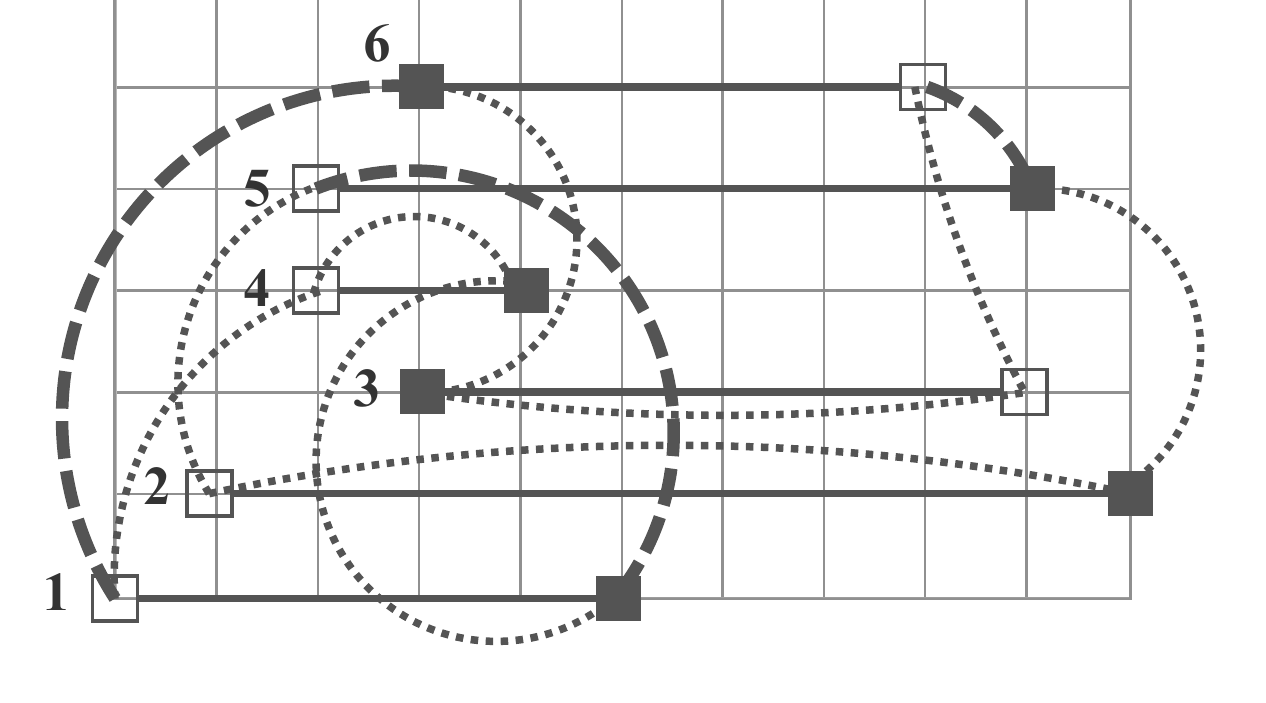}
  \caption{The solution for the PDP problem computed by the PDP algorithm.}
  \label{fig:pdp}
\end{figure}

We now let $SOL$ denote the total Euclidean distance covered by the vehicle for the plan proposed by the PDP algorithm. The optimal solution is denoted by $OPT$. The total length of the delivery-to-pickup segments from $S$ that we use in Step $6$ is $\|S_0\|_d$ where $d$ refers to the dimension of the Euclidean space.

\begin{lemma}
\label{lemma:darp}
\begin{equation}SOL \leq \sum_{i=1}^n \frac{\|s_i - t_i\|_d}{c} + \sqrt{2}\left(\frac{c-1}{c}\right) \|T\|_{2d} + \|S_0\|_d \enspace .\end{equation}
\end{lemma}
\begin{proof}
Compared to Lemma~\ref{lemma:cab} the extra distance driven is $\|S_0\|_d$.
\end{proof}

\begin{observation}
\label{observation:d++}
If the following conditions are met
\begin{equation}
\label{eq:S0}
c\|S_0\|_d = o(n) \enspace ,
\end{equation}
\begin{equation}
\label{eq:T}
c \|T\|_{2d} =  o(n) \enspace , \mbox{ and }
\end{equation}
\begin{equation}
\label{eq:mu}
\sum_{i=1}^n \|s_i - t_i\|_d = \Omega(n) \enspace ,
\end{equation} then the PDP Algorithm is asymptotically optimal: $\lim_{n \rightarrow +\infty} \frac{SOL_n}{OPT_n} = 1$.
\end{observation}

It follows from~\cite{Treleaven2013} that $\|S_0\|_d$ is $o(n)$
almost surely\footnote{The $o(n)$ result follows from (2) and the
unnumbered equation in the proof of Theorem 4.5 in
\cite{Treleaven2013}.} for $d \geq 3$ if $S$ is the Stacker Crane
tour computed by a SPLICE algorithm on $n$ requests in $[0,1]^d
\times [0,1]^d$ with the origins and destinations picked i.i.d.
using the same distribution. The Stacker Crane tour $S$ from our PDP
algorithm consists of $n/c$ requests but the corresponding points
are not picked independently. Informally speaking, these $n/c$
requests seem to be spread evenly on $R$ so we are optimistic with
respect to proving that (\ref{eq:S0}) holds but more work has to be
done to look into the details and conditions for convergence.
Observation~\ref{observation:d++} gives us reason to believe that
our PDP algorithm is asymptotically optimal for fixed capacity $c$
if the origins and destinations are picked i.i.d. from the same
distribution since (\ref{eq:T}) and (\ref{eq:mu}) hold in this case
and (\ref{eq:S0}) seems plausible.

\subsection*{Conclusion}

We have presented a polynomial time asymptotically optimal algorithm
for the PDP-C problem and a polynomial time algorithm for the PDP
problem that we have good reasons to believe is asymptotically
optimal as well (under certain assumptions for picking the
transportation requests). Our results are dealing with vehicles with
limited capacity greater than one. One obvious idea for future work
is incorporating time windows by extending the dimension of the
request space.

\bibliographystyle{plain}
\bibliography{Olsen}

\begin{thebibliography}{10}

\bibitem{Arora1998}
S.~Arora.
\newblock Polynomial time approximation schemes for euclidean traveling
  salesman and other geometric problems.
\newblock {\em J. ACM}, 45(5):753--782, September 1998.

\bibitem{BeardwoodHaltonHammersley1959}
J.~Beardwood, J.~H. Halton, and J.~M. Hammersley.
\newblock The shortest path through many points.
\newblock {\em Mathematical Proceedings of the Cambridge Philosophical
  Society}, 55(4):299–327, 1959.

\bibitem{Berbeglia2007}
G.~Berbeglia, J.~F. Cordeau, I.~Gribkovskaia, and G.~Laporte.
\newblock Static pickup and delivery problems: a classification scheme and
  survey.
\newblock {\em TOP: An Official Journal of the Spanish Society of Statistics
  and Operations Research}, 15(1):1--31, 2007.

\bibitem{christofides76}
N.~Christofides.
\newblock {Worst-case analysis of a new heuristic for the travelling salesman
  problem}.
\newblock Technical Report 388, Graduate School of Industrial Administration,
  Carnegie Mellon University, 1976.

\bibitem{Das15}
A.~Das and C~Mathieu.
\newblock A quasipolynomial time approximation scheme for euclidean capacitated
  vehicle routing.
\newblock {\em Algorithmica}, 73(1):115--142, 2015.

\bibitem{few1955}
L.~Few.
\newblock The shortest path and the shortest road through n points.
\newblock {\em Mathematika}, 2(2):141--144, 12 1955.

\bibitem{GUAN1998}
D.J. Guan.
\newblock Routing a vehicle of capacity greater than one.
\newblock {\em Discrete Applied Mathematics}, 81(1):41 -- 57, 1998.

\bibitem{Haimovich_Kan_1985}
M.~Haimovich and A.~H. G.~Rinnooy Kan.
\newblock Bounds and heuristics for capacitated routing problems.
\newblock {\em Mathematics of Operations Research}, 10(4):527--542, 1985.

\bibitem{Khachay2016}
M.~Khachay and R.~Dubinin.
\newblock {P}{T}{A}{S} for the euclidean capacitated vehicle routing problem in
  ${R}^d$.
\newblock In {\em DOOR 2016}, volume 9869 of {\em Lecture Notes in Computer
  Science}, pages 193--205. Springer, 2016.

\bibitem{Parragh2008_1}
S.~N. Parragh, K.~F. Doerner, and R.~F. Hartl.
\newblock A survey on pickup and delivery problems (part {I}).
\newblock {\em Journal f\"{u}r Betriebswirtschaft}, 58(1):21--51, 2008.

\bibitem{Parragh2008_2}
S.~N. Parragh, K.~F. Doerner, and R.~F. Hartl.
\newblock A survey on pickup and delivery problems (part {I}{I}).
\newblock {\em Journal f\"{u}r Betriebswirtschaft}, 58(2):81--117, 2008.

\bibitem{Psaraftis1981}
{H. N.} Psaraftis.
\newblock Analysis of an o(n²) heuristic for the single vehicle many-to-many
  euclidean dial-a-ride problem.
\newblock {\em Transportation Research. Part B: Methodological}, pages
  133--145, 1981.

\bibitem{Savelsbergh95}
M.~W.~P. Savelsbergh and M.~Sol.
\newblock The general pickup and delivery problem.
\newblock {\em Transportation Science}, 29:17--29, 1995.

\bibitem{Stein1978}
D.~M. Stein.
\newblock An asymptotic, probabilistic analysis of a routing problem.
\newblock {\em Mathematics of Operations Research}, 3(2):89--101, 1978.

\bibitem{Treleaven2013}
K.~Treleaven, M.~Pavone, and E.~Frazzoli.
\newblock Asymptotically optimal algorithms for one-to-one pickup and delivery
  problems with applications to transportation systems.
\newblock {\em IEEE Transactions on Automatic Control}, 58(9):2261--2276, 2013.

\end{thebibliography}

\end{document}